\documentclass[foot, reqno]{amsart}
\usepackage{amsaddr}
\usepackage[backend=biber,style=numeric,sorting=nyt, natbib=true]{biblatex}
\usepackage[english]{babel}
\usepackage{csquotes}
\usepackage{amsmath,amssymb,latexsym}
\usepackage{graphicx}
\usepackage{xcolor}
\usepackage{float}
\usepackage{multicol}
\usepackage{subcaption}
\usepackage{enumerate}
\usepackage[shortlabels]{enumitem}
\usepackage{makecell}
\usepackage{booktabs}

\usepackage{xurl}
\usepackage{hyperref}

\newtheorem{theorem}{Theorem}[section]
\newtheorem{lemma}[theorem]{Lemma}

\theoremstyle{definition}

\theoremstyle{remark}
\newtheorem{remark}[theorem]{Remark}

\begin{document}
\title{Gradient-free ensemble transform methods for generalized Bayesian inference in generative models}

\author{Diksha Bhandari$^{1*}$, Sebastian Reich$^1$}
\thanks{$^{*}$Corresponding author: Diksha Bhandari}
\email{diksha.bhandari@uni-potsdam.de}
\email{sebastian.reich@uni-potsdam.de}
\address{$^1$University of Potsdam,
Institute of Mathematics, Karl-Liebknecht Str. 24/25, D-14476 Potsdam, Germany}
\thanks{The research has been funded by the Deutsche Forschungsgemeinschaft (DFG)- Project-ID 318763901 - SFB1294.}




\begin{abstract}
Bayesian inference in complex generative models is often obstructed by the absence of tractable likelihoods and the infeasibility of computing gradients of high-dimensional simulators. Existing likelihood-free methods for generalized Bayesian inference typically rely on gradient-based optimization or reparameterization, which can be computationally expensive and often inapplicable to black-box simulators. To overcome these limitations, we introduce a gradient-free ensemble transform Langevin dynamics method for generalized Bayesian inference using the maximum mean discrepancy. By relying on ensemble-based covariance structures rather than simulator derivatives, the proposed method enables robust posterior approximation without requiring access to gradients of the forward model, making it applicable to a broader class of likelihood-free models. The method is affine invariant, computationally efficient, and robust to model misspecification. Through numerical experiments on well-specified chaotic dynamical systems, and misspecified generative models with contaminated data, we demonstrate that the proposed method achieves comparable or improved accuracy relative to existing gradient-based methods, while substantially reducing computational cost. 
\end{abstract}

\maketitle

\section{Introduction}
\label{sec:lfi_intro}
In modern machine learning and statistical inference, many generative models of interest are defined only implicitly through simulators rather than through explicit probability densities. For such models, it is typically straightforward to generate independent and identically distributed (i.i.d.) samples from the distribution \( P_{\theta} \) for any given parameter value \( \theta \in \Theta \), but the corresponding likelihood function is not available in closed form. This lack of a tractable likelihood renders classical inference techniques, such as maximum likelihood estimation or Bayesian inference via posterior sampling difficult or even infeasible to apply directly.

To address this challenge, a growing body of work has focused on likelihood-free or simulator-based inference methods, which replace likelihood evaluations with comparisons between simulated and observed samples. Within this family, methods based on the maximum mean discrepancy (MMD) have emerged as particularly powerful. The MMD is a kernel-based statistical distance that quantifies discrepancies between probability distributions and can be estimated efficiently using only samples. When employed in a frequentist setting, the MMD forms the basis of minimum distance estimators that are consistent, asymptotically normal, and robust to model misspecification \citep{briol_statistical_2019}. Recently, MMD-based approaches have also been introduced in the Bayesian framework through the MMD-Bayes posterior, a scoring rule-based formulation that enables Bayesian inference without the need for likelihood evaluations \citep{cherief-abdellatif_mmd-bayes_2020, pacchiardi_generalized_2024}. Existing work has primarily relied on gradient-based techniques, such as stochastic gradient Langevin dynamics (SGLD), to explore this posterior \citep{pacchiardi_generalized_2024}. However, these methods require access to gradients of the simulator or rely on reparameterization techniques, which are often unavailable or impractical for complex, real-world models.

In this paper, we build on two complementary strands of literature: (i) generalized, loss-based Bayesian inference, and (ii) ensemble Kalman filter techniques for sampling. In particular, affine-invariant methods such as the affine invariant interacting Langevin dynamics (ALDI) \citep{garbuno-inigo_affine_2020} and the ensemble Kalman sampler (EKS) \citep{garbuno-inigo_interacting_2020, ding_ensemble_2021} have been widely studied as interacting particle systems that approximate posterior distributions. However, most existing formulations of these methods are restricted to standard \( \ell_2 \)-based loss functions.
We extend this line of work by introducing a novel, \emph{gradient-free ensemble transform Langevin dynamics} method for generalized Bayesian inference using the MMD loss. The proposed framework targets the MMD-Bayes posterior while avoiding explicit likelihood gradients, thereby enabling robust inference for a wide range of intractable or misspecified models. Our approach relies on interacting particle systems driven by Langevin-type flows that approximate the desired posterior through a combination of sample-based updates and ensemble interactions. Furthermore, while many loss-based inference methods are sensitive to affine transformations of the model parameters \citep{matsubara_robust_2022}, our proposed method is affine invariant, thereby achieving greater stability and interpretability across reparameterizations of \( \theta \).

\subsection{Loss-based generalized Bayesian inference}
\label{sec:gen_bayes}
In classical Bayesian inference, the likelihood function represents all the information that the observed data provides about the unknown parameters \( \theta \). It formalizes the assumed data-generating process and serves as a central component in updating prior beliefs to form the posterior distribution via Bayes' rule,
\begin{equation}
\pi_{\text{post}}(\theta) \propto \pi_{\text{prior}}(\theta) \, \pi(y \mid \theta),
\label{eqn:bayes3}
\end{equation}
where \( \pi_{\text{prior}}(\theta) \) represents prior knowledge about the parameters, and \( \pi(y \mid \theta) \) denotes the likelihood describing the probabilistic relationship between the model output \( G(\theta) \) and the observed data \( y \). Consider a family of probability distributions
\[
\mathcal{P}_{\Theta}(\mathcal{X})
=
\left\{
P_{\theta} \text{ defined on } \mathcal{X} \subset \mathbb{R}^N
\;:\;
\theta \in \Theta \subset \mathbb{R}^D
\right\}.
\]
In the \emph{well-specified} case, the true data-generating distribution \( P_{\text{data}} \) belongs to this family, meaning that there exists a true parameter value \( \theta^* \in \Theta \) such that \( P_{\theta^*} = P_{\text{data}} \). If the model admits a tractable density \( \pi(\cdot, \theta) \), the maximum likelihood estimator (MLE) can be obtained as
\[
\theta_{\text{MLE}}^* = \arg\max_{\theta \in \Theta} \, \frac{1}{N} \sum_{n=1}^N \log \pi(y^n \mid \theta),
\]
for i.i.d. samples \( \{y^n\}_{n=1}^N \sim P_{\text{data}} \). In this setting, as the number of observations grows, the Bayesian posterior converges almost surely to a multivariate normal distribution centered at the true parameter \( \theta^* \), with covariance matching the asymptotic sampling covariance of the MLE, as described by the Bernstein–von Mises theorem \citep{ghosh_introduction_2006}. In practice, however, many realistic models, particularly simulator-based or likelihood-free models do not admit a closed form or differentiable likelihood. Moreover, the assumption that \( P_{\text{data}} \) lies within the parametric family \( \mathcal{P}_{\Theta}(\mathcal{X}) \) often fails. Approximate forward models, contaminated or heterogeneous data, and structural model errors can lead to \emph{model misspecification}. In such cases, there is no true parameter value, and the Bayesian posterior asymptotically concentrates around the parameter minimizing the Kullback–Leibler (KL) divergence to the true data distribution,
\[
\theta^*_{\text{KL}} = \arg\min_{\theta} \, \mathrm{KL}(P_{\theta}, P_{\text{data}}).
\]
While this provides a principled asymptotic limit, the resulting posterior can exhibit poor robustness properties, particularly in the presence of outliers or heavy tailed noise, since the KL divergence heavily penalizes tail discrepancies.

To address these challenges, \emph{loss-based generalized Bayesian inference} provides a broader and more flexible framework. Instead of requiring a likelihood function, it defines the posterior distribution through an alternative loss or scoring rule that quantifies the discrepancy between observed data and model predictions. Formally, the generalized Bayesian posterior \citep{bissiri_general_2016} is given by 
\begin{equation}\label{eqn:gen_bayes_posterior}
\pi_{\text{post}}(\theta \mid y) \propto \pi_{\text{prior}}(\theta) \, \exp\!\big(-w \, l(y, \theta)\big),
\end{equation}
where \( l(y, \theta) \) is a chosen loss function, and \( w \) is a scaling parameter  balancing the influence of prior information and observed data. Under suitable regularity conditions on \( l \), the generalized posterior satisfies a Bernstein–von Mises result, remaining asymptotically normal and centered at the minimizer of the loss function as the sample size increases \citep{miller_asymptotic_2021, pacchiardi_generalized_2024}.

The flexibility of this formulation lies in the choice of the loss function \( l \), which can be tailored to the problem at hand. Common examples include robust scoring rules such as the energy score, kernel score \citep{gneiting_strictly_2007}, or maximum mean discrepancy, each of which enhances robustness to model misspecification and outliers. When \( l(y, \theta) = -\log \pi(y \mid \theta) \) and \( w = 1 \), the generalized posterior \eqref{eqn:gen_bayes_posterior} reduces to the standard Bayesian posterior in \eqref{eqn:bayes3}, corresponding to the KL divergence. Choosing \( w \neq 1 \) with the same log-loss results in the fractional posterior \citep{holmes_assigning_2017, bhattacharya_bayesian_2019}. Overall, loss-based generalized Bayesian inference extends the Bayesian paradigm to complex settings where likelihoods are unavailable, or infeasible to evaluate, offering a principled and robust alternative for modern likelihood-free inference.

\subsection{Related work}
Kernel-based methods have become increasingly popular in statistical inference due to their nonparametric flexibility and strong theoretical guarantees. Among these, the maximum mean discrepancy has emerged as a particularly powerful tool for comparing probability distributions via reproducing kernel Hilbert space (RKHS) embeddings \citep{gretton_kernel_2006}. Over the past decade, MMD has been widely adopted in likelihood-free inference, inspiring both frequentist estimators and Bayesian formulations. From a frequentist perspective, \citet{briol_statistical_2019} provided a rigorous analysis of the statistical properties of minimum-MMD estimators, establishing results on consistency and asymptotic normality. They further introduced a minimum-distance estimation framework, where parameters are estimated by minimizing the MMD between simulated model outputs and observed data. Their work also proposed a natural gradient descent-type algorithm for efficiently computing these estimators, making MMD a practical and theoretically grounded approach for inference without explicit likelihoods.

Building upon these ideas, \citet{cherief-abdellatif_mmd-bayes_2020} proposed a Bayesian analogue known as MMD-Bayes, in which the classical likelihood term in Bayes' rule is replaced by a kernel-based pseudo-likelihood. This formulation enables Bayesian inference in a wide variety of generative models where likelihoods are unavailable. One of the main challenges in MMD-Bayes is that the resulting posterior is analytically intractable. To overcome this, \citet{cherief-abdellatif_mmd-bayes_2020} developed a variational inference approach, along with a stochastic gradient algorithm to approximate the MMD-Bayes posterior. More recently, \citet{pacchiardi_generalized_2024} extended the MMD-Bayes framework by placing it within the broader context of generalized Bayesian inference based on scoring rules. Their formulation employs adaptive stochastic gradient Langevin dynamics and pseudo-MCMC schemes \citep{nemeth_stochastic_2021, jones_adaptive_2011, andrieu_pseudo-marginal_2009} to sample from the resulting posteriors. While effective, these methods depend on kernel gradients estimated from finite samples and require extensive simulation at each iteration for a proposed sample, which can be computationally demanding. The generalized Bayes formulation followed in \citep{pacchiardi_generalized_2024} builds on the foundational work of \citet{bissiri_general_2016} and \citet{jewson_principles_2018}, in which the likelihood is replaced by an arbitrary loss or scoring rule that measures the discrepancy between model predictions and observed data. Within this framework, the MMD-Bayes method naturally corresponds to choosing the kernel score as the scoring rule. Related approaches such as Bayesian Synthetic Likelihood (BSL) \citep{price_bayesian_2018} can also be interpreted within the same generalized Bayesian formulation. A complementary line of work has explored alternative loss-based Bayesian formulations using other kernel discrepancies. For example, \citet{matsubara_robust_2022} investigated generalized Bayesian inference based on the kernel Stein discrepancy \citep{chwialkowski_kernel_2016, liu_kernelized_2016}, showing that the resulting posterior distributions enjoy both consistency and robustness properties. Unlike our focus on fully likelihood-free models, their approach applies to doubly intractable models, where the likelihood is known only up to a normalizing constant. Additionally, \citet{dellaporta_robust_2022} proposed a nonparametric approach to robust Bayesian inference by introducing an MMD posterior bootstrap, which combines the posterior bootstrap with MMD-based estimators to achieve robustness against model misspecification.

Several other methods have been developed to perform Bayesian inference without explicit likelihood evaluations. Approximate Bayesian computation (ABC) \citep{del_moral_adaptive_2012, lintusaari_fundamentals_2017, park_k2-abc_2016} compares observed and simulated summary statistics directly, but its accuracy strongly depends on the choice of summaries and distance metrics. More recent developments include neural approaches such as neural density ratio estimation \citep{thomas_likelihood-free_2022} and variational methods based on synthetic likelihoods, both of which aim to scale likelihood-free inference to complex, high-dimensional models.

\subsection{Outline and contributions of this work}
This paper introduces a novel gradient-free, affine-invariant ensemble transform method for generalized Bayesian inference that employs the maximum mean discrepancy (MMD) as a loss function in likelihood-free settings. The proposed approach builds upon the principles of ensemble-based Langevin dynamics to enable efficient and robust posterior approximation without requiring likelihood evaluations or simulator gradients. By conducting extensive numerical experiments on a range of generative models, we show that the method delivers robust and accurate inference not only in well-specified scenarios but also under misspecification, such as when observations are contaminated or when dealing with chaotic dynamical systems. Furthermore, systematic comparisons with standard Bayesian inference and gradient-based Langevin dynamics show that our method attains comparable, and in many cases superior, estimation accuracy, while offering improved robustness to model misspecification.

The remainder of this paper is organized as follows.  
Section~\ref{sec:definitions} introduces the definition of the MMD and its empirical estimators, along with the formulation of the minimum-MMD estimator.  
In Section~\ref{sec:lfi}, we describe the use of MMD within the framework of generalized Bayesian inference and present the resulting MMD-Bayes formulation for likelihood-free inference.  
Section~\ref{sec:etld_methods} develops the proposed gradient-free gradient-free ensemble transform Langevin dynamics method in detail, including its derivation, affine-invariance properties, and practical implementation.  
Finally, Section~\ref{sec:experiments} provides empirical evaluations across various benchmark problems and compares the performance of the propose method with alternative inference approaches.

\section{Framework and definitions}
\label{sec:definitions}
\subsection{Notation}
We denote by \(\mathcal{X} \subset \mathbb{R}^N\) the data space and by \(\Theta \subset \mathbb{R}^D\) the parameter space. Let \(\mathcal{P}(\mathcal{X})\) denote the set of Borel probability measures on \(\mathcal{X}\). A generic probability measure is denoted by \(Q \in \mathcal{P}(\mathcal{X})\), and its density with respect to Lebesgue measure is denoted by \(\pi_Q(\cdot)\), when it exists.

We consider a family of generative models indexed by parameter \(\theta \in \Theta\), defined via a family of probability measures 
$ \mathcal{P}_{\Theta}(\mathcal{X}) = \{P_\theta \in \mathcal{P}(\mathcal{X}) : \theta \in \Theta\}.$
When \(P_\theta\) has a density with respect to the Lebesgue measure, we denote it as \(\pi(\cdot \mid \theta)\).

The true data-generating distribution is denoted by \(P_{\text{data}}\), and its density (if available) by \(\pi_{\text{data}}(\cdot)\). Even if the likelihood is not accessible, we can sample i.i.d. realizations $\{y^n\}_{n=1}^N \sim P_{\text{data}}$. 

\subsection{Generative model setting}
Let \((\mathcal{U}, \mathcal{F}, \mathbb{U})\) be a probability space that serves as a latent source of randomness. A generative model $\mathcal{P}_{\Theta}(\mathcal{X})$ is defined as a family of probability measures parameterized by $\Theta$, such that, for any $\theta \in \Theta$, we can sample i.i.d. realizations $\left\{x^j\right\}_{j=1}^J \subset \mathcal{X}$ from $P_\theta$. These realizations are obtained via a function \(G_\theta: \mathcal{U} \to \mathcal{X},
\) called the \emph{generator} or \emph{simulator}, which maps latent samples to data space. Thus, the probability distribution induced by \(G_\theta\) on \(\mathcal{X}\) is the pushforward measure \(P_\theta = G_{\theta\#}\mathbb{U}\). 
To generate \(n\) independent samples from the model, we first draw i.i.d. samples \(\{u^j\}_{j=1}^J \sim \mathbb{U}\), and apply the generator to each of these samples:
$$
x^j = G_\theta(u^j), j = 1, \dots, J.
$$
This formulation enables simulation from the generative model even when the likelihood \(\pi(x \mid \theta)\) is intractable or does not exist. Our goal is to identify a parameter \(\theta^*\) such that the corresponding model \(P_{\theta^*}\) is closest to \(P_{\text{data}}\), in some sense of ‘closeness' defined by a loss function or statistical criterion. 

Further, to set the foundation for our likelihood-free inference framework, we discuss the key mathematical concepts that underpin our methodology. Specifically, we define the maximum mean discrepancy, a widely used kernel-based divergence measure for comparing probability distributions, along with its empirical formulations. We also present the minimum MMD estimator, which identifies parameters that minimize the discrepancy between model-generated and observed data distributions. These tools form the basis for the generalized Bayesian framework studied in the subsequent sections.

\subsection{Maximum mean discrepancy}
\label{sec:mmd}
First introduced in \citep{gretton_kernel_2006}, the maximum mean discrepancy is a popular kernel-based statistical measure used to compare probability distributions. In the context of likelihood-free inference, MMD provides a powerful tool to quantify the discrepancy between a model generated distribution and observed data. 

Let $\mathcal{H}_k$ be a reproducing kernel Hilbert space (RKHS) with associated positive-definite and symmetric kernel $k: \mathcal{X} \times \mathcal{X} \rightarrow \mathbb{R}$. Let $\mathcal{F}$ be a class of functions $f : \mathcal{X} \rightarrow \mathbb{R}$, then the maximum mean discrepancy between any two probability distributions measures \(P, Q \in \mathcal{P}(\mathcal{X})\) relative to $\mathcal{F}$ is defined as
\begin{equation}
    \text{MMD}_{\mathcal{F}}(P, Q) := \sup_{f \in \mathcal{F}} \left[ \mathbb{E}_{x \sim P}[f(x)] - \mathbb{E}_{y \sim Q}[f(y)] \right].
\end{equation}
In general, the MMD is a pseudo-metric for probability distributions (i.e., symmetric, satisfies the triangle inequality and $\text{MMD}_{\mathcal{F}}(P,P) = 0$ \citep{briol_statistical_2019}). When $\mathcal{F}$ is the unit ball $\mathcal{F} := \{f \in \mathcal{H}_k : \|f\|_{\mathcal{H}_k} \leq 1 \}$ in $\mathcal{H}_k$, it allows to avoid computing the supremum explicitly and is convenient to compute. Due to the reproducing properties $(f(x) = \langle f, k(x, \cdot) \rangle_{\mathcal{H}_k})$ of the RKHS, the MMD can be expressed as the distance between the kernel mean embeddings of the two distributions as
\begin{equation}\label{eqn:mmd}
    \operatorname{MMD}(P, Q) = \left\| \mu_P - \mu_Q \right\|_{\mathcal{H}_k},
\end{equation}
where \(\mu_P := \mathbb{E}_{x \sim P}[k(x, \cdot)]\) and \(\mu_Q := \mathbb{E}_{y \sim Q}[k(y, \cdot)]\) are the kernel mean embeddings of \(P\) and \(Q\) in $\mathcal{H}_k$, respectively. 
It has been shown that on $\mathbb{R}^D$, Gaussian and Laplace kernels are characteristic, so the MMD for these two kernels is a metric for distributions on $\mathbb{R}^D$ \citep{gretton_kernel_2012}. Due to these properties, in this work, we will focus on the Gaussian kernel for MMD 
$$k(x,y) = \exp{\left(-\frac{\left\| x-y \right\|_2^2}{2 \gamma^2} \right)},$$
where $\gamma$ is scalar bandwidth. The squared-MMD is particularly useful in applications, due to its simple expression given by \citep{gretton_kernel_2012}
\begin{align}
    \operatorname{MMD}^2(P, Q)
    &= \left\| \mu_P - \mu_Q \right\|_{\mathcal{H}_k}^2 \notag \\
    & =
\mathbb{E}_{x, x' \sim P}\left[k(x, x')\right]
+
\mathbb{E}_{y, y' \sim Q}\left[k(y, y')\right]
-
2\, \mathbb{E}_{x \sim P, y \sim Q}\left[k(x, y)\right].\label{eqn:mmd_squared}
\end{align}
In practice, we do not have direct access to the distributions \( P \) or \( Q \), but only to samples from them.
Given i.i.d. samples \(\{x^j\}_{j=1}^J \sim P_\theta\) simulated from the model, and \(\{y^n\}_{n=1}^N \sim P_{\text{data}}\) drawn from the true data distribution, we use the following unbiased empirical estimate of \(\mathrm{MMD}^2(P_\theta, P_{\text{data}})\):
\begin{align} \label{eqn:mmd2_empirical}
\widehat{\operatorname{MMD}}^2(P_\theta, P_{\text{data}})
&= \mathbb{E}_{x, x' \sim P_{\theta}}\left[k(x, x')\right]
+
\mathbb{E}_{y, y' \sim P_{\text{data}}}\left[k(y, y')\right] - 2\, \mathbb{E}_{x \sim P_{\theta}, y \sim P_{\text{data}}}\left[k(x, y)\right] \notag \\ 
& = \frac{1}{J(J-1)} \sum_{\substack{j, j' = 1 \\ j \neq j'}}^J k(x^j, x^{j'})
+
\frac{1}{N(N-1)} \sum_{\substack{n, n' = 1 \\ n \neq n'}}^N k(y^n, y^{n'}) \notag \\
&\quad
-
\frac{2}{JN} \sum_{j=1}^J \sum_{n=1}^N k(x^j, y^n).
\end{align}

\subsection{Minimum MMD Estimators}
\label{sec:min_mmd_estimator}
Given a parameter set $\theta \in \Theta$, the minimum MMD estimator, as defined in \citep{briol_statistical_2019}, is the parameter value that minimizes the squared MMD between a parameterized family of models \( P_{\theta} \) and unknown data generating distribution \( P_{\text{data}} \):
\begin{equation}
    \hat{\theta}_{\text{MMD}} := \arg\min_{\theta \in \Theta} \, \operatorname{MMD}^2(P_{\theta}, P_{\text{data}}).
\end{equation}
for $\{y^n\}_{n=1}^N \sim P_{\text{data}}$ i.i.d. In the next section, we focus on a Bayesian framework for likelihood-free inference based on the MMD estimator.

\section{Likelihood-free inference using maximum mean discrepancy}
\label{sec:lfi}
Standard Bayesian inference traditionally requires access to a likelihood function \( \pi(\cdot \mid \theta) \) to update prior beliefs about a parameter \( \theta \). However, in many complex models, particularly those defined by simulators, the likelihood is either unavailable or computationally intractable. To address such challenges, generalized Bayesian inference provides a flexible alternative by allowing updates to the prior using alternative data-fit objectives, such as general loss functions or proper scoring rules \citep{bissiri_general_2016, pacchiardi_generalized_2024}. Under the assumed likelihood-free setting, given a loss function \( l(y, \theta) \), the generalized Bayesian posterior is defined as 
\begin{equation}\label{eqn:gen_bayes}
\pi_{\text{post}}(\theta \mid y) \propto \pi_{\text{prior}}(\theta) \exp\big(-w ~ l(y, \theta)\big),
\end{equation}
where \( w > 0 \) is a learning rate. More broadly, it allows Bayesian updating even in the absence of a likelihood function by using alternative discrepancy measures. 

We often have access to a generator \( G_{\theta}(\cdot) \), which defines a parameterized model distribution \( P_{\theta} \), from which we can easily simulate. Hence a suitable choice of $l(y, \theta)$ would be for which estimators can be obtained with the generated samples. A particularly useful class of losses arises from distributional discrepancy measures. One such measure is the maximum mean discrepancy, defined in Section \ref{sec:mmd}, which quantifies the difference between two distributions using RKHS embeddings. As shown in \citep{cherief-abdellatif_mmd-bayes_2020}, the MMD-Bayes posterior inherits desirable properties: it is consistent and robust to model misspecification, even in cases where the standard (likelihood-based) posterior would fail to concentrate on the true parameters. Thus, even under moderate model misspecification, minimizing the MMD can recover the model distribution closest to the data-generating process, resulting in robust inference. Given an unknown true data generating distribution \( P_{\text{data}} \), and a parameterized model distribution \( P_{\theta} \), then the squared MMD between $P_{\theta}$ and $ P_{\text{data}}$ is given by \eqref{eqn:mmd2_empirical}, for a positive-definite kernel $k$.
Substituting the $\mathrm{MMD}^2$ into the generalized Bayes formula \eqref{eqn:gen_bayes} results in the following MMD-Bayes posterior
\begin{equation}\label{eqn:mmd_bayes}
\pi_{\text{MMD}}(\theta \mid Y) 
\propto \pi_{\text{prior}}(\theta) \exp\left(-\beta ~ \mathrm{MMD}^2(P_\theta, P_{\text{data}})\right),
\end{equation}
where $Y=\{y^n\}_{n=1}^N \sim  P_{\text{data}}$, and \( \beta > 0\) is a temperature parameter that adjusts the influence of the data-fit term relative to the prior. This posterior concentrates around parameter values for which the generated distribution \( P_\theta \) is closest to the observed distribution \( P_{\text{data}} \) in MMD. This framework is especially appealing in likelihood-free inference, where models are implicitly defined via simulators rather than explicit likelihoods. Since MMD-Bayes only requires the ability to simulate from \( P_\theta \) and to evaluate a kernel function \( k \), it is broadly applicable in practice. 
Note that the choice of $\beta$ in MMD-Bayes posterior \eqref{eqn:mmd_bayes} represents the amount of the influence of the observations relative to the prior. For a standard Bayes update, in a well-specified case, $\beta$ is fixed to $1$ \citep{zellner_optimal_1988}. For misspecified models, several works \citep{grunwald_inconsistency_2017, holmes_assigning_2017} have investigated the use of $\beta \neq 1$, with various criteria proposed for its selection \citep{bissiri_general_2016, lyddon_general_2019, matsubara_robust_2022}. In this work, we adopt the technique of \citep{bissiri_general_2016}, also implemented in \citep{pacchiardi_generalized_2024}, which is computationally tractable and does not require samples from the exact posterior. We do not pursue an optimal $\beta$ selection here, noting that it remains an active research area best addressed in a problem specific manner.
\begin{remark}
The work by \citep{pacchiardi_generalized_2024} proposes using a sample-based scoring rule, denoted by \( S(P_\theta, y) \), as the loss function \( l(y, \theta) \) in the generalized Bayes framework \eqref{eqn:gen_bayes}. To perform inference, they employ MCMC techniques targeting the resulting scoring rule posterior. Since exact sampling from the scoring rule posterior is not possible, they utilize a pseudo-marginal Monte Carlo Markov chain (PM-MCMC) approach, where for each proposed \( \theta^\dagger \), simulations from $P_{\theta^{\dagger}}$ are used to sample from a close approximation. However, reducing the error of this approximation requires an increasing number of simulations per MCMC step, resulting in significant computational cost. This limits the practicality of the method in complex models where simulations are expensive.
\end{remark}
As highlighted before, the framework of loss-based generalized Bayesian inference is not limited to choice of MMD as a discrepancy measure, but can be extended to a more general framework with scoring rules. Setting \( l(y, \theta) \) as a proper scoring rule \( S(P_\theta, y) \), one can obtain a scoring rule posterior $\pi_{\text{S}}(\theta \mid y) $, introduced in \citep{pacchiardi_generalized_2024}. In particular, inserting the kernel score in \eqref{eqn:gen_bayes}, recovers the MMD-Bayes posterior. For detailed proofs and discussion on the asymptotic normality, finite sample generalization bound, and a quantitative bound on the robustness to outliers of the kernel score posterior (and thus MMD-Bayes posterior), we refer the reader to \citep{pacchiardi_statistical_2022}. 

In the following section, we propose a novel gradient-free method for sampling from the MMD-Bayes posterior \eqref{eqn:mmd_bayes}, based on ensemble transform Langevin dynamics. By leveraging the ensemble structure, this approach enables efficient posterior exploration without relying on explicit gradient information or incurring the high computational cost typically associated with large simulation budgets. 

\section{Ensemble transform methods for likelihood-free inference}
\label{sec:etld_methods}
Recent advances in likelihood-free inference have explored gradient-based Monte Carlo and variational techniques for sampling from the MMD-Bayes posterior $\pi_{\text{MMD}}$ in scenarios where the likelihood is intractable \citep{pacchiardi_generalized_2024, cherief-abdellatif_mmd-bayes_2020}. In this work, we build on these developments by introducing an ensemble-based approach rooted in overdamped Langevin dynamics. For time $s >0$, the central idea is to evolve an ensemble of particles $\{\theta^{m}_s\}_{m=1}^M$ to approximate the posterior as $s \rightarrow \infty$, using information derived entirely from simulation outputs and without requiring gradient access to the forward model $G_{\theta_s}(\cdot)$. At time $s=0$, let $\{\theta^{m}_0\}_{m=1}^M$ denote an ensemble of $M$ parameters drawn from the prior distribution $\pi_{\text{prior}}$. We sample random noise input $\{u^{j}\}_{j=1}^J \sim \mathbb{U}$, with corresponding model outputs transformed by the generator $G_{\theta_s^m}$ as
 \begin{equation*}
     x^{mj}_s = G_{\theta^m_s}(u^j), \qquad j=1,\dots,J,
  \end{equation*}
for each $\theta^m_s$, $m=1,\dots, M$. Each $\theta^{m}_s \in \mathbb{R}^D$ represents a parameter vector in a $D$-dimensional space, while $\{x^{mj}_s\}_{j=1}^J$ are the corresponding realizations from the parameterized model. Throughout this work, we will adapt the convention that upper indices stand for the ensemble index or the $i$th entry of an $N$-dimensional vector, while lower indices stand for the time index. We start by defining the empirical means and covariances as 
\begin{align}
 \bar{\theta}_s  & = \frac{1}{M} \sum_{m=1}^M \theta^{m}_s, \quad \bar{x}_s^j  = \frac{1}{M} \sum_{m=1}^M x^{mj}_s, \label{eq:means} \\ 
 C_s & = \frac{1}{M} \sum_{m=1}^M (\theta^{m}_s - \bar{\theta}_s) (\theta^{m}_s - \bar{\theta}_s)^\mathrm T, \label{eqn:c_theta} \\
C^{\theta x^j}_s &= \frac{1}{M} \sum_{m=1}^M (\theta^{m}_s - \bar{\theta}_s) (x^{mj}_s - \bar{x}_s^j)^\mathrm T. \label{eqn:c_theta_x}
\end{align}
In the following section, we present a gradient-free interacting particle system which evolves particles $\{\theta^{m}_s\}_{m=1}^M$ over an infinite time horizon to get samples from the approximate posterior $\pi_{\text{MMD}}$ defined in \eqref{eqn:mmd_bayes}.

\subsection{Gradient-free ensemble transform Langevin dynamics}
\label{sec:gf_etld}
In this section, we present an interacting particle system of based on overdamped Langevin dynamics, and its gradient-free approximation for application in generalized Bayesian inference for likelihood-free models. Our approach is inspired by the affine-invariant interacting Langevin dynamics (ALDI) framework \cite{garbuno-inigo_affine_2020}, where empirical covariance structures act as preconditioner to the Langevin dynamics. This design improves efficiency by adapting to the local geometry of the posterior. While the work in \cite{garbuno-inigo_affine_2020} focuses on the $l_2$-loss to sample from the target posterior, we extend this approach for application to general loss functions, specifically the MMD, to sample from the MMD-Bayes posterior. For time $s > 0$, define the ensemble as $\{\theta^{m}_s\}_{m=1}^M$, and associated model outputs  $\{x^{mj}_s\}_{j=1}^J$ for $m = 1, \dots, M$. We introduce a matrix of ensemble deviations 
\begin{equation}
    \theta' = (\theta^1 - \bar{\theta}_s, \theta^2 - \bar{\theta}_s, \ldots, \theta^J - \bar{\theta}_s)
\end{equation}
Thus, rewriting the empirical covariance $C_s$ in a compact form
\begin{equation}
    C_s = \frac{1}{J} \theta' \theta'^{~\rm T}.
\end{equation}
Motivated by the work in \citep{garbuno-inigo_affine_2020}, we choose the $C_s$ as the preconditioning matrix. Thus, for particles $\theta^{m}_s$, and a potential $\mathcal{V}(\theta^{m}_s) \in \mathbb{R}^{D \times M}$, we consider the following type of finite-size IPS 
\begin{equation}\label{eqn:aldi_sde} 
\mathrm{d}\theta_s^{m} = - C_s\nabla_{\theta^m} ~\mathcal{V}(\theta_s^{m})~ \mathrm{d}s + \frac{D+1}{M}(\theta_s^m-\bar{\theta}_s)\, \mathrm{d}s + \sqrt{2}C^{1/2}_s\, \mathrm{d}W_s^{m},
\end{equation}
where $W_s^{m}$ denotes a $D$-dimensional Brownian motion.
\begin{remark}
Following \citep{garbuno-inigo_affine_2020}, in the large ensemble limit $(M>D)$, a generalized non-symmetric square root of covariance matrix $C_s$ can be defined using ensemble deviations as
$$C^{1/2}_s= \frac{1}{\sqrt{J}}\theta'$$
i.e., $C_s = C_s^{1/2}(C_s^{1/2})^{\mathrm T}.$ This generalized square root is particularly useful in high-dimensional settings where computing a Cholesky factorization of $C_s$ can be computationally prohibitive. 
\end{remark}
\begin{remark}
    For a fixed $D$, the correction term $\frac{D+1}{M}(\theta_s^m-\bar{\theta}_s)$ in \eqref{eqn:aldi_sde} is needed to sample from the desired target measure for a small number of particles $M$ (i.e., $M < D$). Moreover, this correction term vanishes as $M \rightarrow \infty$. For a detailed derivation of this interacting particle system, and its mean-field limit, we refer the reader to \citep{garbuno-inigo_affine_2020, nusken_note_2019}.
\end{remark}
For likelihood-free inference with MMD, we choose the potential $\mathcal{V}(\theta_s^{m})$ in \eqref{eqn:aldi_sde} as the negative
log-posterior
\begin{equation}\label{eqn:mmd_potential}
\mathcal{V}(\theta_s^{m}) = -\log \pi_{\text{prior}}(\theta^m) + \beta \, \mathrm{MMD}^2(P_{\theta^m_s}, P_{\text{data}}),
\end{equation}
for $m = 1, \dots, M$. We use the shorthand $\mathrm{MMD}^2(P_{\theta^m_s}, P_{\text{data}}) := \mathrm{MMD}^2(\theta^m_s)$. Thus, the resulting IPS SDE is 
\begin{align}\label{eqn:mmd_aldi_sde}
\mathrm{d}\theta_s^{m} =& - C_s \left(\nabla_{\theta^m} \log \pi_{\text{prior}}(\theta^m)- \beta\nabla_{\theta^m}\mathrm{MMD}^2(\theta^m_s)\right)\, \mathrm{d}s \notag \\ & + \frac{D+1}{M}(\theta_s^m-\bar{\theta}_s) \, \mathrm{d}s + \sqrt{2} C^{1/2}_s\, \mathrm{d}W_s^{m}
\end{align}
for $m = 1, \dots, M,$ where $W_s^{m}$ denotes a $D$-dimensional standard Brownian motion. This evolution requires the evaluation of $\nabla_{\theta^m}\mathrm{MMD}^2(\theta^m_s)$. This expression requires the Jacobian of the forward model \(G_{\theta_s^m}(u^j)\). Computing this Jacobian is often impractical or computationally infeasible for complex simulators, even in the case of a single particle trajectory $\theta_s$. Therefore, to overcome this challenge, we propose a gradient-free approximation for Bayesian inference with MMD, using ensemble covariance statistics, inspired by the EKS \citep{ding_ensemble_2021}, and ALDI frameworks \citep{garbuno-inigo_affine_2020}. These methods approximate parameter-to-output relationships via ensemble-based covariance estimates instead of directly computing model derivatives. As shown in \citep{weissmann_gradient_2022, schillings_analysis_2016}, for general non-linear forward operators, we can express the cross-covariance between parameters and outputs via
\begin{equation}\label{eqn:jacobian_surrogate}
C^{\theta x^j}_s
\approx
C_s \, \nabla_{\theta^m}G_{\theta_s^m}(u^j),
\end{equation}
for $m =1, \dots,M$, where $C^{\theta x^j}_s \in \mathbb{R}^{D \times N}$ is the sample cross-covariance between the parameter ensemble and the model outputs \eqref{eqn:c_theta_x}, and $C_s \in \mathbb{R}^{D \times D}$ is the parameter covariance matrix as defined in \eqref{eqn:c_theta}. 
\begin{remark}
    The approximation \eqref{eqn:jacobian_surrogate} is referred to as \emph{statistical linearization} and can be used to replace the derivative within any sampling (or optimization) method for Bayesian inference. This allows the conversion of even single particle sampling methods for inverse problems, with dependence on the derivative of the forward model into a derivative-free interacting particle system for sampling. Moreover, as shown in \citep{calvello_ensemble_2025, garbuno-inigo_affine_2020}, the statistical linearization ansatz preserves the affine invariance of the IPS \eqref{eqn:mmd_aldi_sde} and is exact in the linear setting, i.e., when $G$ is a linear map then $C^{\theta x^j}_s = C_s \, \nabla_{\theta^m}G_{\theta_s^m}(u^j)$. 
\end{remark}
We introduce an approximation of $\nabla_{\theta^m} \widehat{\operatorname{MMD}}^2(\theta^m_s)$. To compute the gradient with respect to a $\theta^m$, we differentiate only those terms in the empirical estimate $\widehat{\mathrm{MMD}}^2$ defined in \eqref{eqn:mmd2_empirical} that are dependent on $\theta^m$. The gradient can be written as
\begin{align} \label{eqn:grad_mmd2_empirical}
\nabla_{\theta^m} \widehat{\operatorname{MMD}}^2(\theta^m_s)
&=
\frac{2}{J(J-1)} \sum_{\substack{l,j = 1 \\ l \neq j}}^J \nabla_{\theta^m} G_{\theta^m_s}(u^j) \nabla_{x^{mj}} k(x^{mj}_s, x^{ml}_s)\notag \\ & \quad -
\frac{2}{JN}  \sum_{j=1}^J \sum_{n=1}^N  \nabla_{\theta^m} G_{\theta^m_s}(u^j) \nabla_{x^{mj}} k(x^{mj}_s, y^n).
\end{align}
Note that a U-statistic approximation for the gradient \eqref{eqn:grad_mmd2_empirical}, over a single particle $\theta$ has been studied in \citep{briol_statistical_2019}, where the gradient is an average over all unordered pairs of samples. Substituting the approximation \eqref{eqn:jacobian_surrogate} in \eqref{eqn:grad_mmd2_empirical}, we have
\begin{align}\label{eqn:approximate_grad_mmd2}
\nabla_{\theta^m} \widehat{\operatorname{MMD}}^2(\theta^m_s)
&=
C_s^{-1}
\Bigg(
\frac{2}{J(J-1)} \sum_{\substack{l,j = 1 \\ l \neq j}}^J  C^{\theta x^j}_s\nabla_{x^{mj}} k(x^{mj}_s, x^{ml}_s) \notag \\
& \quad -
\frac{2}{J N}  \sum_{j=1}^J  \sum_{n=1}^N  C^{\theta x^j}_s \nabla_{x^{mj}} k(x^{mj}_s, y^n)
\Bigg) \notag \\
& = C_s^{-1} ~ g^{mj}_s,
\end{align}
where 
\begin{equation}\label{eqn:g_cross_cov}
g^{mj}_s := \frac{2}{J(J-1)} \sum_{\substack{l,j = 1 \\ l \neq j}}^J C^{\theta x^j}_s\nabla_{x^{mj}} k(x^{mj}_s, x^{ml}_s)
- \frac{2}{JN}  \sum_{j=1}^J  \sum_{n=1}^N C^{\theta x^j}_s\nabla_{x^{mj}} k(x^{mj}_s, y^n).
\end{equation}
Combining the resulting gradient approximation \eqref{eqn:approximate_grad_mmd2} with the Langevin dynamics \eqref{eqn:mmd_aldi_sde}, results in the gradient-free ensemble transform Langevin dynamics (GF-ETLD) formulation given by the interacting particle system 
\begin{align}\label{eqn:grad_free_ips}
\mathrm{d}\theta_s^{m} =& - \Big(C_s \nabla_{\theta^m} \log \pi_{\text{prior}}(\theta^m) - \beta ~ g^{mj}_s\,\Big) \mathrm{d}s \notag  \\ 
&  + \frac{D+1}{M}(\theta_s^m-\bar{\theta}_s) \, \mathrm{d}s + \sqrt{2} ~C^{1/2}_s\, \mathrm{d}W_s^{m} ,
\end{align}
for $m = 1, \dots, M,$ where $W_s^{m}$ denotes a $D$-dimensional standard Brownian motion, and $g_s^{mj}$ is defined in \eqref{eqn:g_cross_cov}. With a fixed step size $\Delta s > 0$, we use the Euler-Maruyama method to discretize the formulation \eqref{eqn:grad_free_ips}. This defines a gradient-free Langevin-type IPS which targets the posterior distribution \eqref{eqn:mmd_bayes} as $s \rightarrow \infty$.
Following the discussion in \citep{garbuno-inigo_affine_2020}, we now prove that the gradient-free ensemble transform Langevin dynamics \eqref{eqn:grad_free_ips} is invariant under affine transformations of the model parameters. Affine invariance in this context means that, under any invertible affine reparameterization of the parameter space, the particle system retains its form and evolution law.

\begin{lemma}[Affine invariance of GF-ETLD] The interacting particle-system \eqref{eqn:mmd_aldi_sde} and its gradient-free formulation \eqref{eqn:grad_free_ips} are affine invariant. That is, under the invertible affine change of variables $\theta = A  \tilde{\theta}  + b$, the gradient-free IPS for $\{ \tilde{\theta} _s^m\}_{m=1}^M$ follows
\begin{align}
\mathrm{d}\tilde{\theta}_s^{\,m} &= - \tilde{C}_s \nabla_{ \tilde{\theta} ^m} \log \tilde{\pi}_{\mathrm{prior}}(\tilde{\theta}^m) \,\mathrm{d}s 
\ - \ \beta\, \, \tilde{g}^{mj}_s \,\mathrm{d}s  \nonumber \\
&\quad + \frac{D+1}{M} \big(\tilde{\theta}_s^m - \bar{ \tilde{\theta} }_s \big) \,\mathrm{d}s 
+ \sqrt{2} \, \tilde{C}^{1/2}_s \, \mathrm{d}W_s^m,
\label{eqn:gfetld-v-affine}
\end{align}
\end{lemma}
\begin{proof}
Let $\theta \in \mathbb{R}^D$ be the model parameter vector, and consider an affine transformation of the form
\begin{equation}
    \theta = A  \tilde{\theta}  + b, 
\end{equation}
for any invertible $A \in \mathbb{R}^{D \times D}$, and any shift $b \in \mathbb{R}^D$. The induced PDFs in $ \tilde{\theta} $-coordinates are defined as
\begin{equation}
    \tilde{\pi}_{s}( \tilde{\theta} ) := |\det A|~ \pi_{s}(A  \tilde{\theta}  + b) .
\end{equation}
For particles $\{\theta_s^m\}_{m=1}^M$, we define the transformed ensemble $\{\tilde{\theta}_s^m\}_{m=1}^M$ by $\tilde{\theta}_s^m := A^{-1}(\theta_s^m - b)$.
Thus, we have
$$C_s = A \tilde{C}_s A^{\mathrm T} \quad \text{and} \quad C^{1/2}_s = A  \tilde{C}^{1/2}_s.$$
Moreover, for functions $\tilde{h}(\tilde{\theta}_s^m) = h(\theta_s^m) = h(A\tilde{\theta}_s^m+b)$, we get
\begin{equation}
    \nabla_{ \tilde{\theta} ^m}~\tilde{h}(\tilde{\theta}_s^m)=  \nabla_{ \tilde{\theta} ^m}h(A\tilde{\theta}_s^m+b)= A^{\mathrm T} \nabla_{\theta^m}h(\theta_s^m) 
\end{equation}
Furthermore, 
\begin{equation}
    C^{\theta x^j}_s = A ~C^{ \tilde{\theta}  x^j}_s,
\end{equation}
for $j=1,\dots,J$. Substituting these in \eqref{eqn:grad_free_ips}, we get \eqref{eqn:gfetld-v-affine}, hence establishing the affine invariance of proposed GF-ETLD method. The result also holds for the gradient-based IPS \eqref{eqn:mmd_aldi_sde}.
\end{proof}
In Section \ref{sec:experiments}, we present a series of numerical experiments that illustrate the method's effectiveness in well-specified models, its robustness under model misspecification, and its computational advantages over existing gradient-based techniques.

\section{Numerical experiments}
\label{sec:experiments}
In this section, we present a set of numerical experiments designed to investigate the performance of the proposed gradient-free ensemble transform Langevin dynamics method for generalized Bayesian inference in both well-specified and misspecified settings. We consider three different types of generative models to assess the performance, robustness, and computational efficiency of the method. First, we evaluate the method on a stochastic dynamical system where the explicit likelihood is intractable. We then proceed to misspecified settings involving synthetic location models, both Gaussian and uniform, to assess robustness against data contamination. Across all cases, we benchmark the performance of gradient-free ensemble transform methods against existing baselines.

\subsection{Well-specified case}
This section focuses on the well-specified setting, which allows us to investigate the method under idealized conditions and isolate the performance of the GF-ETLD algorithm without the added challenge of model misspecification. For this purpose, we choose the  stochastic version of the Lorenz96 model due to its chaotic dynamics, high dimensionality, and practical relevance in atmospheric modeling. This experiment serves as a benchmark for assessing the method's ability to accurately recover latent parameters in complex dynamical systems.
\subsubsection{Stochastic Lorenz96 model}
In this experiment, we investigate the application of the GF-ETLD method for generalized Bayesian inference on a parameterized stochastic Lorenz96 model. The Lorenz96 system \citep{hagedorn_predictability_2006} is a classical testbed in data assimilation and dynamical systems, often used to model nonlinear chaotic dynamics found in atmospheric systems, including slow and fast variables. Here, we use a modified stochastic version, introduced in \citep{wilks_effects_2005}, where the effect of fast variables on the slow variables is approximated by a stochastic effective term depending on a set of parameters. Here, the true likelihood is unaccessible, thus, sampling from the exact posterior is not possible.

We consider the stochastic variant of the Lorenz96 model given by the following system of coupled differential equations
\begin{equation}\label{eqn:stochastic_lorenz96}
\frac{dy_k}{dt} = -y_{k-1}(y_{k-2} - y_{k+1}) - y_k + F - g_k, \quad \text{for } k = 1, \dots, K,
\end{equation}
where \( y_k(t) \) is the state of the \(k\)-th variable at time \(t\), and the indices are cyclic, i.e., $y_{K+1} = y_1$, $y_{K+2}= y_2$, and so on. The stochastic forcing term \(g(y_k, t; \theta)\) depends on a parameter vector \( \theta = (b_0, b_1, \sigma_e, \phi) \), models the impact of unresolved fast-scale processes on each slow variable $y_k(t)$ and is defined upon discretizing the differential equations with step-size $\Delta t$ as
\begin{equation}
g(y_k, t; \theta) = b_0 + b_1 y_k + \phi \, r(t - \Delta t) + \sigma_e \sqrt{1 - \phi^2} \, \eta(t),
\end{equation}
where \( \theta = (b_0, b_1, \phi, \sigma_e) \) are unknown parameters. The residual process \( r(t) \) follows a discrete-time autoregressive process of order one (AR(1)) given by
\begin{equation}
r(t) = \phi r(t - \Delta t) + \sigma_e \sqrt{1 - \phi^2} \, \eta(t),
\label{eq:ar1}
\end{equation}
where \( \eta(t) \sim \mathcal{N}(0,1) \) is i.i.d. standard Gaussian noise, \( \phi \) is the autoregressive coefficient, and \( \sigma_e > 0 \) controls the magnitude of the stochastic term. The stochastic forcing enters each slow variable equation identically but independently through the shared AR(1) process, thus coupling the noise with the state dynamics.

Our goal is to infer parameters $\theta = (b_0, b_1, \phi, \sigma_e) \in \mathbb{R}^4$. We fix $K = 8$ variables, choose $F=10$, and numerically integrate the model using 4th order Runge-Kutta method over a time interval $[0, T]$, with $T = 2.5$, using a fixed time step $\Delta t = 3/40$, resulting in $K = T/ \Delta t = 33$ time steps, and data dimension $33 \times 8 = 264$. We simulate from the stochastic Lorenz96 model using the true parameters $\theta^* = (2.0, 0.8, 0.9, 1.7)$ to obtain the observed data from this system. Our goal is to sample from the MMD-Bayes posterior over $\theta$ given the observed trajectory using the GF-ETLD method. We initialize ensemble of particles $\{\theta^m\}_{m=1}^M$ from the prior $\pi_{\text{prior}} = \mathcal{N}([1.0, 0.0, 0.0, 1.0], (2\operatorname{I}, \operatorname{I}, 2\operatorname{I}, \operatorname{I}))$ and run the GF-ETLD sampler with $M = 200$, and step size $\Delta s = 1e{-3}$. 

To quantitatively assess the accuracy of parameter inference, we compute the root mean square error (RMSE) between the posterior mean and the true parameter values for each of the four parameters \( \theta = (b_0, b_1, \phi, \sigma_e) \). We compare the RMSE obtained using the proposed GF-ETLD method against the adaptive stochastic gradient Langevin dynamics (adSGLD) approach based on the kernel score, as proposed by \cite{pacchiardi_generalized_2024}. This method uses stochastic gradients of a kernel scoring rule to perform posterior sampling. Table \ref{tab:lorenz96_rmse} summarizes the RMSE for both methods and parameters. The results demonstrate that GF-ETLD achieves lower or comparable RMSE to adSGLD across all parameters, highlighting its performance in the likelihood-free setting of the stochastic Lorenz96 model. However, due to its dependence on gradient evaluation of the generator, the adSGLD method is computationally more intensive, a point we examine further in Section \ref{sec:computational_time}.


\begin{table*}[t]
\begin{tabular}{@{}lcccc@{}}
\hline
\textbf{Method} & \multicolumn{1}{c}{\textbf{\(b_0\)}}
& \multicolumn{1}{c}{\textbf{\(b_1\)}} & \multicolumn{1}{c}{\textbf{\(\phi\)}}
& \multicolumn{1}{c}{\textbf{\(\sigma_e\)} }  \\
\hline
GF-ETLD         & 0.0081 & 0.0011 & 0.0021 & 0.0199 \\
adSGLD          & 0.0023 & 0.0016 & 0.0035 & 0.0074 \\
\hline
\end{tabular}
\caption{RMSE of posterior means for each parameter in the stochastic Lorenz96 model.}
\label{tab:lorenz96_rmse}
\end{table*}
These results highlight the effectiveness of GF-ETLD for efficient likelihood-free inference on stochastic dynamical systems, particularly in scenarios where gradients of the simulator are unavailable or computationally prohibitive. This experiment underscores the method's practical suitability as an inference tool for dynamical systems characterized by differential equations with embedded noise processes and intractable likelihoods.

\subsection{Misspecified case}
We now investigate the more challenging case of model misspecification, where the observed data does not perfectly conform to the assumed data generating process. This is typical in real-world settings, where data contamination and distributional shifts are common. Our goal is to assess the robustness of GF-ETLD under such conditions. We use two small scale yet insightful generative models: (i) a Gaussian location model and (ii) a uniform location model, to illustrate the performance of GF-ETLD in the presence of varying levels of data contamination. We examine whether the method can maintain accurate inference even when the likelihood is misaligned with the true data distribution.

\subsubsection{Gaussian models as generative models} 
We evaluate the performance of the proposed gradient-free ensemble transform Langevin dynamics method for sampling from the MMD-Bayes posterior. This experiment focuses on a univariate \textit{normal location model} where the mean parameter $\theta$ is unknown and the standard deviation $\sigma = 1$ is fixed.

To generate data, we define the forward model $G_\theta(u) = \theta + u$, where $u \sim \mathcal{N}(0,1)$, so that $x^j = G_\theta(u^j) \sim \mathcal{N}(\theta, 1)$, for $j=1,\dots,J$. We examine robustness to model misspecification by considering data contamination. Specifically, $N = 150$ samples are generated such that a proportion $1 - \epsilon$ are drawn from $\mathcal{N}(\theta^*, 1)$ with true $\theta^* = 0$, and a proportion $\epsilon$ are drawn from $\mathcal{N}(z, 1)$ with $z = 10$. The contamination level $\epsilon \in \{0, 0.1, 0.2\}$ controls the degree of outlier presence in the data.

We compare three inference methods on this problem. First, we compute the standard Bayesian posterior using the exact likelihood and a Gaussian prior $\theta \sim \mathcal{N}(2,1)$. This serves as a baseline for ideal conditions but is expected to degrade with increase in contamination. Second, we use the adSGLD approach based on the kernel score \cite{pacchiardi_generalized_2024}. Finally, we apply our GF-ETLD sampler, where an ensemble of parameters $\{\theta^m_0\}_{m=1}^M$ is initialized from the prior and evolved using the interacting Langevin dynamics derived in Section \ref{sec:gf_etld}. We use $M = 10$ particles, a step size of $\Delta s = 10^{-3}$, and evolve the system for $12$ time steps.

Figure \ref{fig:ms_location_posteriors} shows the resulting posterior densities across the different outlier contamination levels. We observe that the GF-ETLD method remains robust across increasing contamination levels. As $\epsilon$ increases from 0.0 to 0.2, the GF-ETLD posteriors continue to concentrate around the true mean $\theta^* = 0$. This demonstrates the ability of the method to downweight the influence of outliers in the data-generating process. The adSGLD method performs comparably to GF-ETLD but yields slightly more dispersed posterior distributions. Moreover, adSGLD incurs higher computational cost, as it requires an increasing number of forward simulations per iteration--a detail addressed in the following section. In contrast, the standard Bayesian posterior, is highly sensitive to contamination. Even at a mild to moderate contamination level of $\epsilon = 0.1, 0.2$, the posterior moves significantly away from the true value, and fails to capture the true parameter, reflecting its sensitivity to model misspecification. These results highlight the value of using ensemble-based methods in robust likelihood-free inference.

\begin{figure}
    \centering
    \includegraphics[width=\textwidth]{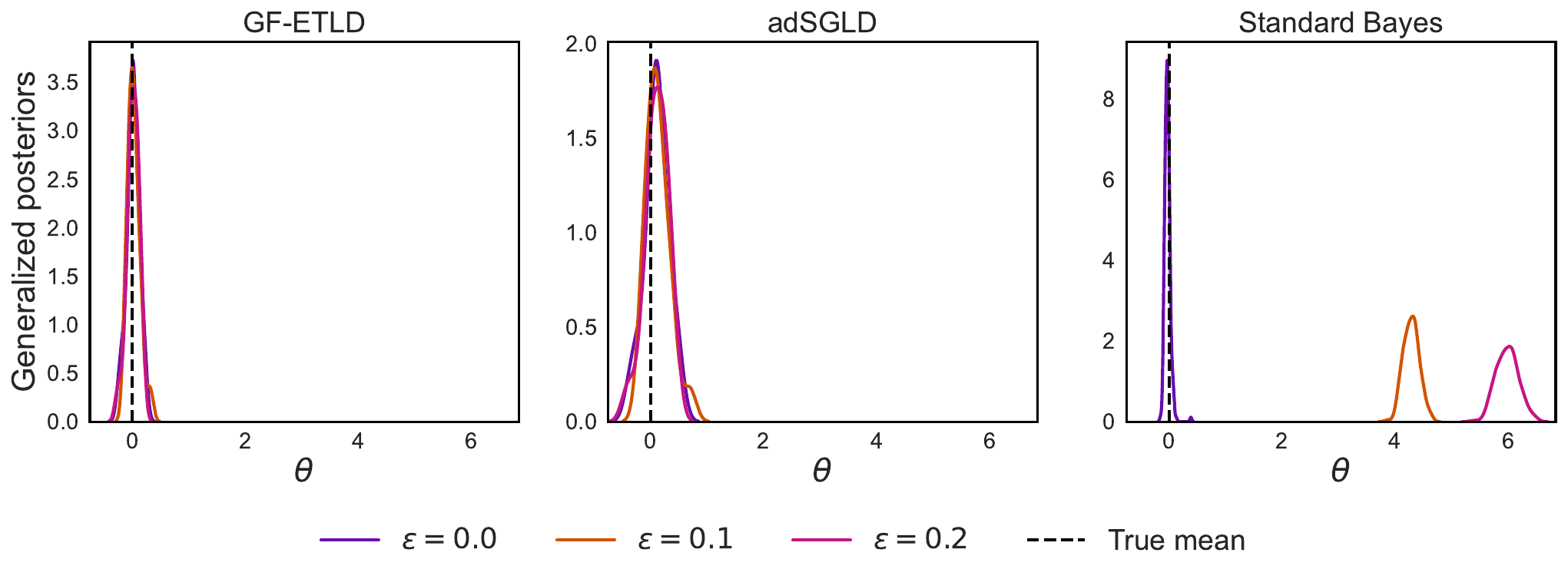}
    \caption[Generalized posteriors of the misspecified normal location model under increasing contamination levels]{Posterior densities under increasing contamination levels $\epsilon=0.0,0.1,0.2$. The vertical dashed line indicates the true parameter value $\theta^* = 0$. Left: GF-ETLD (MMD-Bayes), middle: adSGLD (kernel-score), right: standard Bayes posterior.}
    \label{fig:ms_location_posteriors}
\end{figure}
The robustness of the GF-ETLD method is further demonstrated in Figure \ref{fig:ms_gf_loc_rmse}, which displays the root mean squared error (RMSE) between the posterior mean and the true location parameter \(\theta^* = 0\) as a function of the outlier proportion \(\epsilon\). We observe that the RMSE remains low and nearly constant for contamination levels up to \(\epsilon = 0.5\), indicating that the GF-ETLD posterior accurately estimates the true parameter even in the presence of a substantial proportion of outliers. Only when the outlier proportion exceeds \(50\%\) does the estimation error increase noticeably, with a steep rise as \(\epsilon\) approaches \(1\). This finding confirms that the GF-ETLD method performs robust inference under moderate contamination.

\begin{figure}
    \centering
    \includegraphics[width=0.8\linewidth]{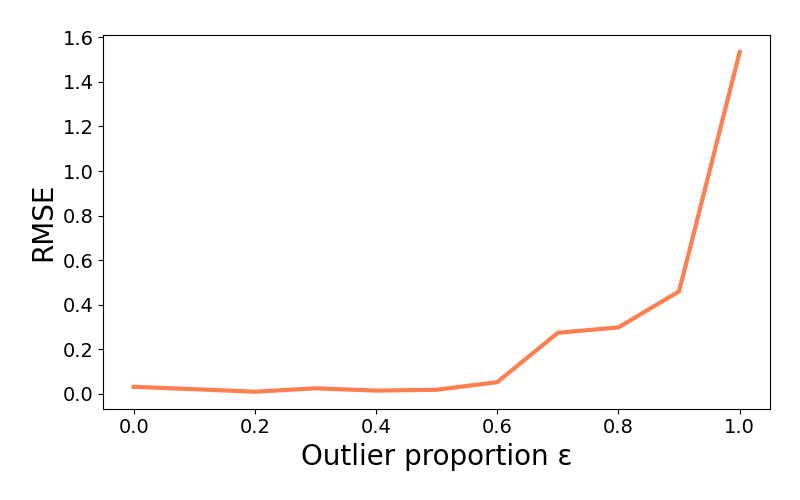}
    \caption[RMSE of GF-ETLD posterior mean in Gaussian location model as a function of outlier proportion]{Root mean squared error of the posterior mean obtained by GF-ETLD for the misspecified Gaussian location model, as a function of the outlier proportion \(\epsilon\).}
    \label{fig:ms_gf_loc_rmse}
\end{figure}
In summary, this experiment highlights the robustness of GF-ETLD to increasing outlier contamination. The method maintains posterior concentration around the true parameter, even under substantial contamination, outperforming standard Bayes and matching adSGLD with lower computational cost. This makes GF-ETLD particularly well-suited for real-world scenarios.

\subsubsection{Uniform location parameter estimation} 
We further examine the behavior of GF-ETLD in a uniformly distributed location model. This setting deviates further from the Gaussian assumptions typically used in Bayesian inference. In this experiment, we estimate the location parameter \(\theta\) of a uniform distribution \(\mathcal{U}([\theta - 1.0, \theta + 1.0])\), with the true value \(\theta^* = 1\). We generate \(N = 100\) i.i.d. observations, and following the misspecified setting in the previous experiment, we introduce contamination by replacing a proportion \(\epsilon \in [0, 0.4]\) of the samples with outliers drawn from a Gaussian distribution \(\mathcal{N}(10, 1)\). We compare the performance of the gradient-free ensemble transform Langevin dynamics method with the standard Bayesian posterior by computing the RMSE between the posterior mean and the true parameter \(\theta^*\). Since the adSGLD method yields results comparable to GF-ETLD, we do not report it explicitly in this setting. 
\begin{figure}
    \centering
    \includegraphics[width=0.8\linewidth]{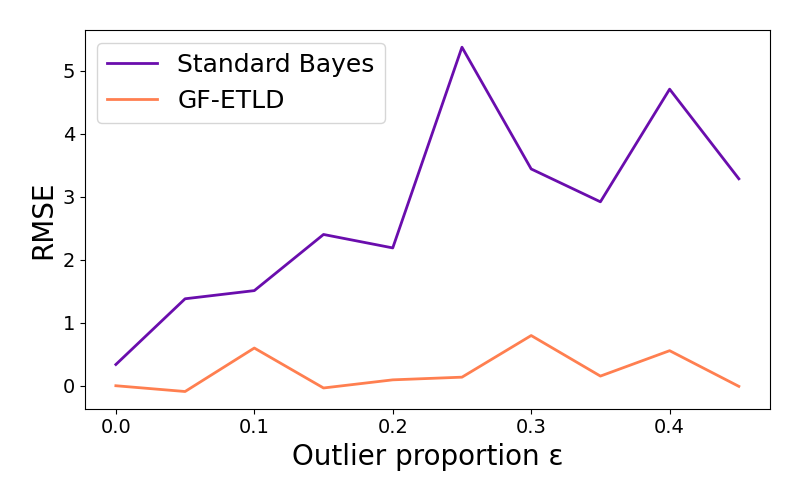}
    \caption[RMSE of GF-ETLD, and standard Bayes posterior estimates in uniform location model under increasing outlier contamination]{Root mean squared error of posterior estimates obtained by GF-ETLD, and standard Bayes approach for the uniform location parameter under increasing proportions of contamination with outliers $\epsilon$}
    \label{fig:ms_uniform_rmse}
\end{figure}
Figure \ref{fig:ms_uniform_rmse} displays the RMSE across increasing values of \(\epsilon\). We observe that the standard Bayesian posterior is highly sensitive to even small proportions of outliers, exhibiting a sharp rise in RMSE. In contrast, the GF-ETLD method maintains significantly lower error throughout, demonstrating robustness to contamination and highlighting its effectiveness in misspecified settings.

This experiment reinforces the robustness advantage of GF-ETLD. The consistent performance of GF-ETLD in both Gaussian and uniform misspecified settings affirms its general applicability to a broad class of likelihood-free inference problems.

\subsection{Computational time}
\label{sec:computational_time}
Finally, we assess the computational efficiency of the GF-ETLD method in comparison with adSGLD and standard Bayesian inference. This section will show that GF-ETLD not only achieves competitive or superior inference accuracy and robustness, but also does so at a lower computational cost. This property is crucial in high-dimensional or simulation-based models where gradient evaluations or large ensemble sizes are computationally burdensome. 

We provide the computational time recorded for each experiment over multiple independent runs. The Gaussian location, and uniform location estimation experiments were repeated $L= 10$ times, and the stochastic Lorenz96 experiment $5$ times. For each run, a new dataset was generated and $L$ posterior samples were obtained for $\theta$. The average time for the generation of $L$ samples is recorded for each method in Table \ref{tab:clock_time}. 

\begin{table*}[t]
\begin{tabular}{@{}lccc@{}}
\hline
\textbf{Method} & \multicolumn{1}{c}{\textbf{Gaussian}}
& \multicolumn{1}{c}{\textbf{Uniform}} & \multicolumn{1}{c}{\textbf{stochastic Lorenz96}} \\
\hline
GF-ETLD         & 1.244 & 0.016 & 1.515 \\
adSGLD          & 8.586 & 6.169 & 55.416 \\
\hline
\end{tabular}
\caption{Average clock time in seconds.}
\label{tab:clock_time}
\end{table*}
The results in Table \ref{tab:clock_time} highlight the computational efficiency of the proposed GF-ETLD method. Across all three experimental setups, GF-ETLD consistently requires significantly less time than adSGLD to generate posterior samples, with improvements spanning one to two orders of magnitude. Consequently, GF-ETLD presents a practical and scalable inference tool for likelihood-free problems, particularly in scenarios where simulator gradients are unavailable, expensive to compute, or where the number of forward simulations must be minimized.

\section{Conclusions}
In this paper, we proposed a gradient-free ensemble transform Langevin dynamics (GF-ETLD) algorithm for generalized Bayesian inference based on the maximum mean discrepancy. By leveraging ensemble-based covariance structures, the approach performs likelihood-free inference without simulator gradients or explicit likelihood evaluations. Empirical studies on the stochastic Lorenz96 system and on contaminated Gaussian and uniform location models demonstrate that GF-ETLD recovers accurate posterior distributions while remaining robust to outliers and model misspecification. Across these benchmarks, it matches or exceeds the accuracy of gradient-based samplers at substantially lower computational cost. These findings position GF-ETLD as a practical and scalable tool for likelihood-free inference in high-dimensional generative models, offering a balance of statistical robustness and computational efficiency. Future research may establish formal theoretical guarantees, explore further scalability via subsampling, localization, or dropout strategies, and extend the framework to a wider range of real-world generative models, including multimodal target distributions, which remains largely unexplored.

\printbibliography

\end{document}